\documentclass[letterpaper, 10pt]{IEEEtran} 
\normalsize

\DeclareMathAlphabet\mathbfcal{OMS}{cmsy}{b}{n}

\usepackage{amsthm}
\usepackage{authblk}
\usepackage{algorithm}
\usepackage{algorithmic}
\usepackage{graphicx,amsmath,amssymb,latexsym,epsfig,url}
\usepackage{setspace}
\usepackage{psfrag}
\usepackage{array,booktabs,arydshln,xcolor,cite}

\usepackage{lscape}
\newtheorem{theorem}{Theorem}

\newtheorem*{proof of Theorem*}{Proof of Theorem 3}
\newtheorem{proof of Lemma}{Proof of Lemma}
\newtheorem{definition}{Definition}

\title{The Optimization of Random Tree Codes for Limited Computational Resources}
\author{B. Tan Bacınoğlu\\
contact: barantan@metu.edu.tr
}

\usepackage{graphicx}

\begin{document}

\maketitle
\date{}

\begin{abstract}
In this paper, we introduce an achievability bound on the frame error rate of random tree code ensembles under a sequential decoding algorithm with a hard computational limit and consider the optimization of the random tree code ensembles over their branching structures/profiles and the decoding measure. Through numerical examples, we show that the achievability bound for the optimizated random tree codes can approach the maximum likelihood (ML) decoding performance of pure random codes. 
\end{abstract}

\section{Introduction}
\label{sec:Introduction}
Random codes have been useful since their use for proving the achievability side of the original channel coding theorem \cite{6773024}. The study of random codes under computationally unconstrained decoding schemes, produced elegant measures for theoretical analysis, and provided achievability bounds as performance benchmarks for practical codes. Contrarily, the characterization of the performance is typically hard for practical codes due to the intractability of the corresponding decoding algorithms. %For example, the complex dynamics of graph-based iterative decoding algorithms like belief propagation lead to approximate analysis techniques like extrinsic information transfer (EXIT) charts.
As such, the performance of practical codes is often truly understood through empirical analysis.

Moreover, conventionally the design of practical codes does not incorporate computational constraints and presumes the success of a low complexity decoder while targetting low error rates. However, one can argue that measures involving the computational constraints on the decoding algorithm will be beneficial as design guidelines for practical codes. One such measure is the computational cutoff rate, or simply the cutoff rate, which estimates the maximum coding rate that the number of backtracing operations can be kept finite during the sequential decoding  \cite{Wozencraft1957SequentialDF} of a tree code \cite{Gallager1968InformationTA}. Polar codes, which are the first practical codes that provably achieve Shannon capacity \cite{5075875}, emerged as a result of a research program dedicated to the boosting of the cutoff rate \cite{DBLP:journals/corr/abs-1908-09594}. In \cite{DBLP:journals/corr/abs-1908-09594}, a practical coding scheme that combines polar codes with convolutional codes, i.e., polar-adjusted convolutional (PAC) codes,  where the sequential decoding of a convolutional code is improved with the channel polarization technique is introduced following the original idea. Notably, PAC codes use convolutional codes as irregular tree codes where the irregularity is due to the irregularity of the rate-profile that comes from the polarized channel.

The aforementioned observations motivates the study of random coding based achievability bounds involving the computational constraints on the decoding algorithm. If the achivability bound is parametrized with design choices, then practical codes can be obtanied using the samples from the random code ensemble that is optimized based on the achievability bound. A significant advantage of this approach is that the study and the design of practical codes can be done without an empirical analysis such as running Monte Carlo simulations. Variational quantum algorithms \cite{Cerezo2020VariationalQA}, such as the quantum approximate optimization (QAOA) algorithm \cite{Farhi2014AQA}, use a similar approach in which solutions to optimization problems are obtained using samples from quantum states prepared with parametrized quantum circuits that are optimized in order to minimize the expectation value of an observable that corresponds to the cost of solutions. For desinging practical codes, an achievability bound involving computational contraints can play a similar role to the expectation of the observable guiding the search for good practical codes. 

In this paper, we provide such an achievability bound for irregular random tree codes under a sequential decoding algorithm with a hard computational constraint. The decoding algorithm that we consider is a variation of the stack algorithm \cite{Zigangirov1975ProceduresOS, Jelinek1969FastSD}, that we refer to as stack-based sequential decoding with give up (SSDGU) algorithm where the decoding algorithm returns an error message in case the number of node checks/visits surpasses a preset computational limit. The memory capacity of the stack is assumed to be as large this computational limit. For a particular class of decoding (error cost) measures that we refer to as accumulating error cost (AEC) measures where the decoding error cost accumulates towards the descendant nodes, the SSDGU algorithm returns either an error message (due to the compuational limit) or the message with the minimum decoding error cost. Using this property, the achievability bound we provide can be  expressed as a sum of two terms: (1) a computatinal limit error (CFE) term and (2) a computational free error (CFE) term.  We consider the optimization of the achievability bound over the parameters that control the branching structure of the irregular random tree code. For this purpose, we suggest an optimization heuristic that we refer to as successive bit placement (SBP) algorithm. As the computational limit grows, the achievability bound optimized with the SBP algorithm attains the minimum value of the CFE part which coincides with the random coding union (RCU) bound \cite{5452208} and Gallager's bound for the relaxed versions of our achievability bound. However, as the computational constraint is finite, the optimization emphasizes computational efficiency. For this reason, we will refer to the optimized ensembles of random tree codes as computationally optimized random tree (CORT) codes. 

The studies in \cite{8630851} and \cite{10619485} are related in targetting approximate maximum likelihood (ML) decoding performance of pure random codes with the use of random codes under practical decoding schemes. In \cite{8630851}, an approximate ML decoding algorithm, i.e., guessing random additive noise decoding with abandonment (GRANDAB), where the decoder searches through an ordered list of noise patterns and abondons the search if a valid codeword is not  found after a fixed number of queries, is introduced. In \cite{10619485}, a class of codes, i.e., random staircase generator matrix codes (SGMCs) are introduces. Staircase-like generator matrices, where only a step descent at time is allowed, considered in \cite{10619485} is a special case  of the generator matrices considered in this paper, up to a change of index ordering\footnote{The random part of the generator matrices in \cite{10619485} is on the lower triangle side while the random part of the generator matrices is on the upper triangle side.}.       In fact, SGMCs correspond to random binary tree codes in our convention. Similarly, the authors  in \cite{10619485} address the computational benefits of SGMCs for approximating ML decoding error rate performance of pure random codes while presenting the representative ordered statistics decoding with local constraints (LC-ROSD) for the decoding of SGMCs.

The rest of the paper is organized as follows: In Section \ref{sec:Descriptions}, we provide the descriptions for the encoding and the decoding schemes that will be consider in the paper. In Section \ref{sec:Main_Results}, we will present our results. In Section \ref{sec:Numerical_Simulations}, we will provide the numerical results showing the evalution of the achievability bound for CORT codes. In Section \ref{sec:Conclusion}, we will make our concluding remarks.  
\section{Descriptions}
\label{sec:Descriptions}
\subsection{Encoding}
We will consider a parametrized ensemble of $(n,k)$-random linear block codes where the ensemble is parametrized through a sequence of integers $\lbrace a_{j} \rbrace_{j=1,2,\cdots ,k }$ with  $a_{j} \in \lbrace 1,2, \cdots , n \rbrace$ representing the arrival time of the message bit $j$. Let $\mathbf{x} \in \lbrace 0,1\rbrace^{n}$ and $\mathbf{m}\in \lbrace 0,1\rbrace^{k}$ denote the coded bits and message bits, respectively. Then, we have: 
\begin{equation}
\mathbf{x}=\mathbf{G}\mathbf{m},
\end{equation}
where $\mathbf{G}$ represents the generator matrix. 

For a randomly picked generator matrix $\mathbf{G}$, $\mathbf{G}_{ij} = 0$ if $i < a_{j}$ and $\mathbf{G}_{ij}$ is independently chosen to be $0$ or $1$ with equal probability if $i \geq a_{j}$. Without loss of generality, we will consider $\lbrace a_{j} \rbrace_{j=1,2,\cdots ,k }$ with non-decreasing order such that $a_{j} \leq a_{\bar{j}}$ for $j < \bar{j}$. For this case, we take $a_{1}=1$ as choosing $a_{1}>1$ fixes coded bits $\mathbf{x}_{i}$, for $i \leq a_{1}$, to zero.

Let $s(t): \lbrace 1, \cdots, n\rbrace \mapsto \lbrace  1, \cdots,k\rbrace$ be the non-decreasing function that corresponds to the number of bits with an arrival time not later than $t$, and can computed as:
\begin{equation}
s(t)= \displaystyle\sum_{j=1}^{k} \mathbb{I}_{\lbrace a_{j} \leq t\rbrace},
\end{equation}
where $\mathbb{I}_{\lbrace \cdot \rbrace}$ is the indicator function.

Note that as we consider $\lbrace a_{j} \rbrace_{j=1,2,\cdots ,k }$ with non-decreasing order, one can obtain $a_{j}$ from $s(t)$ using:
\begin{equation}
a_{j} = \min\lbrace t : s(t)= j\rbrace.
\end{equation}
Also observe that the considered linear block codes are tree codes as we have:
\begin{equation}
\mathbf{x}_{1:t}=\mathbf{G}[1:t;1:s(t)]\mathbf{m}_{1:s(t)},
\end{equation} 
where $\mathbf{G}[1:t;1:s(t)]$ is the submatrix of $\mathbf{G}$ limited to rows $1:t$ and columns $1:s(t)$.

Accordingly, $\mathbf{m}_{1:s(t)}$ as  the subvector of the message bits $\mathbf{m}$ which outputs $\mathbf{x}_{1:t}$, corresponds to one of the $2^{s(t)}$ possible paths for the first $t$ coded bits in the code tree. Let $\varepsilon_{t}: \lbrace 0,1\rbrace^{s(t)} \mapsto \lbrace 0,1\rbrace^{t}$ represent the mapping for the encoding of the first $t$ coded bits, i.e., $\varepsilon_{t}(\cdot)$ is the linear mapping with $\mathbf{G}[1:t;1:s(t)]$ as in below:
\begin{equation}
\varepsilon_{t}(\mathbf{m}_{1:s(t)})=\mathbf{G}[1:t;1:s(t)]\mathbf{m}_{1:s(t)}.
\end{equation} 
Particularly, we will be interested in subvectors of message bits $\mathbf{m}$ at \emph{branching times} which are defined as follows:
\begin{definition}
Given $s(\cdot)$, the time $t$ is said to be a branching time/stage if $s(t)>s(t-1)$ for $t>1$ and $s(t)>0$ for $t=1$.
\end{definition}
Note that as $s(t)>0$, $t=1$ is the first branching time. Let $h_{f}$ represent the last branching time before $n$. Let $b_{h}$ denote the $h$th branching time for $h \leq h_{f}$ and let $b_{h_{f}+1}=n$. Then $b_{h}$ can be expressed as:
\begin{equation}
b_{h}= \min\lbrace t: s(t)>s(b_{h-1})\rbrace,
\end{equation} 
for $1<h \leq h_{f}$.

Any binary sequence of lenght $s(b_{h})$ identifies a node at $h$th branching time/stage in the code tree. Accordingly, we will refer a node at  $h$th branching time using the corresponding binary vector in $\lbrace 0, 1\rbrace^{b_{h}}$. Let $\mathcal{C}(\mathbf{m}_{1:s(b_{h})})$ represent the set of subvectors that correspond to the children of the node $\mathbf{m}_{1:s(b_{h})}$, then it can be expressed as:
\begin{equation}
\mathcal{C}(\mathbf{m}_{1:s(b_{h})})= \lbrace \hat{\mathbf{m}}_{1:s(b_{h+1})}: \hat{\mathbf{m}}_{1:s(b_{h})}=\mathbf{m}_{1:s(b_{h})}\rbrace.
\end{equation} 
In other words, $\mathcal{C}(\mathbf{m}_{1:s(b_{h})})$ is the set of length $s(b_{h+1})$ subvectors that start with subvector $\mathbf{m}_{1:s(b_{h})}$. Let $c_{h}$ represent the number of children nodes of a node at branching time $h$. Let $\mathbf{m}_{0}$ and $c_{0}$ represent the root node and the number of the children nodes of the root node, respectively.

The random tree code considered in this section will be refered as a \emph{$(n,k)$-random tree code with the tree structure $s(t)$}.

\subsection{Decoding}
Let $\mathbf{y} \in \mathcal{Y}^{n}$ represent the output of a binary input noisy channel that takes the coded bits $\mathbf{x}$ as inputs where $\mathcal{Y}$ is the output alphabet of the channel. For the decoding of $\mathbf{m}$ from $\mathbf{y}$, we will consider a version of the sequential decoding with stack, i.e., stack algorithm. The decoding algorithm is associated with a measure $d(\cdot,\cdot): \lbrace 0,1\rbrace^{*}\times\mathcal{Y}^{*}\mapsto \mathbb{R}^{+}$, that represents the decoding error cost for a subvector of coded bits and the corresponding noisy channel outputs, e.g.,  $d(\mathbf{x}_{1:t},\mathbf{y}_{1:t})$, and a computational limit paramater $L$ that represents the maximum number of node checking steps allowed in the algorithm. To express the decoding error cost of a node $\mathbf{m}_{1:\ell}$, we will use the notation $d_{\mid \mathbf{y}}(\mathbf{m}_{1:\ell}):=d(\varepsilon_{r_{\ell}}(\mathbf{m}_{1:s(r_{\ell})}),\mathbf{y}_{1:r_{\ell}})$ where $r_{\ell}=\max\lbrace t \leq n : s(t)=\ell\rbrace$  for more compact expressions. 

We will consider a particular class of measures for $d(\cdot,\cdot)$ that we refer to as \emph{accumulating error cost} (AEC) measures with the following definition:
\begin{definition}
The measure $d(\cdot,\cdot)$ is said to be accumulating error cost (AEC) measure if it satisfies:
\begin{equation}
d(\mathbf{x}_{1:t},\mathbf{y}_{1:t}) \leq d(\mathbf{x}_{1:t'},\mathbf{y}_{1:t'})
\end{equation}
for all $\mathbf{x} \in \mathcal{X}^{n}$, $\mathbf{y} \in \mathcal{Y}^{n}$ and $t < t'$.
\end{definition}
\begin{algorithm}
\caption{The stack-based sequential decoding with give-up (SSDGU) algorithm:}
\begin{algorithmic}
\STATE{$\mathcal{S} \gets \mathcal{C}(\mathbf{m}_{0})$;} \COMMENT{Initialize the stack set $\mathcal{S}$ with the children of the root node.}
\STATE{$N_{c} \gets c_{0}$;} \COMMENT{Initialize the counter $N_{c}$ to the number of children nodes of the root node.}
\WHILE{$N_{c} \leq L$}
\STATE{$\mathbf{\hat{m}} \gets \arg\min_{\mathbf{m'} \in \mathcal{S}} d_{\mid \mathbf{y}}(\mathbf{m'})$;}
\STATE {$\mathcal{S}\gets \mathcal{S} / \lbrace \mathbf{\hat{m}}\rbrace$;} 
\IF{$\mathbf{\hat{m}}$ is a terminal node} \RETURN $\mathbf{\hat{m}}$; \ENDIF
\STATE{$\mathcal{S} \gets \mathcal{S}\cup \mathcal{C}(\mathbf{\hat{m}})$;}
\STATE{$N_{c} \gets N_{c}+ \vert\mathcal{C}(\mathbf{\hat{m}})\vert$;} 
\ENDWHILE
\STATE $\mathbf{\hat{m}} \gets \varepsilon $;
\RETURN $\mathbf{\hat{m}}$;
\end{algorithmic}
\end{algorithm}
As the decoding algorithm, we will consider a variation of the stack sequantial decoding algorithm that we will refer to as the stack-based sequential decoding with give-up (SSDGU) algorithm.
Let $\mathbf{\hat{m}} \in \lbrace 0,1\rbrace^{k} \cup \lbrace \varepsilon \rbrace$ denote the decoded message as the output of the decoding algorithm, where $\varepsilon$ represents the error message returned in case $L$ node checks are made without reaching a terminal node that correponds to a complete message vector with the minimum decoding error cost. Let $\mathcal{S}$ represent the set of node entries in the stack that can store $L$ node entries. We assume that the stack is kept ordered with respect to the decoding error costs of the node entries. Accordingly, the node with the minimum decoding error is the top entry of the stack while finding the correct order of the new entries can be done in $O(\log_{2}(L))$ time using binary search.

Note that the set of descendants of the set of nodes in the stack of the decoding algorithm always covers all the terminal nodes in tree of the code. Due to the definition of AEC measures, this guarantees that the minimum decoding error cost among all the terminal nodes is at least as large as the minimum decoding error cost among the nodes in the stack at any time during the execution of the decoding algorithm. Accordingly, if the decoding algorithm returns a message $\mathbf{\hat{m}}$ in $\lbrace 0,1\rbrace^{k}$, i.e., $\mathbf{\hat{m}} \neq \varepsilon$, then it is guaranteed that the corresponding terminal node has the lowest decoding error cost among all the terminal nodes. 

\section{Main Results}
\label{sec:Main_Results}
Our results will be based on the following result:
\begin{theorem}
\label{DcleDcfeachievability}
If the SSDGU algorithm with an AEC measure $d(\cdot,\cdot)$ and a computational limit paramater $L$ is used to decode an independent random message $\mathbf{m}$ that is uniformly selected from $\lbrace 0,1\rbrace^{k}$ from the output $\mathbf{y}$ of noisy channel that takes the coded bits $\mathbf{x}=\mathbf{G}\mathbf{m}$ with  $\mathbf{G}$ being the generator matrix  of a $(n,k)$-random tree code with the tree structure $s(t)$, then the probability of error for the decoded message $\mathbf{\hat{m}}$ satisfies:
\begin{equation}
\label{theboundforerrorDcleDcfe}
\Pr(\mathbf{\hat{m}}\neq \mathbf{m}) \leq D_{\text{CLE}}[s,d]+D_{\text{CFE}}[s,d],
\end{equation}
where $D_{\text{CLE}}[s,d]$ is the computation-limit-error (CLE) bound which is expressed as
\begin{eqnarray}
&&D_{\text{CLE}}[s,d] = \nonumber\\ 
&&\Pr\left( c_{0}+\displaystyle\sum_{h=1}^{h_{f}-1}c_{h}\displaystyle\sum_{\mathbf{m'}_{1:s(b_{h})}}\mathbb{I}_{\lbrace d_{\mid \mathbf{y}}(\mathbf{m'}_{1:s(b_{h})}) \leq d_{\mid \mathbf{y}}(\mathbf{m})\rbrace} \geq L\right) \nonumber
\end{eqnarray}  
, where the second summation runs over all $\mathbf{m'}_{1:s(b_{h})} \in \lbrace 0,1\rbrace^{s(b_{h})}$ and  $c_{h}=2^{s(b_{h+1})-s(b_{h})}$, $c_{0}=2^{s(b_{1})}$, and $D_{\text{CFE}}[s,d]$ is the computation-free-error (CFE) bound which is defined as
\begin{eqnarray}
&&D_{\text{CFE}}[s,d] = \nonumber\\
&&\Pr\left( \displaystyle\sum_{\mathbf{m'}}\mathbb{I}_{\lbrace d_{\mid \mathbf{y}}(\mathbf{m'}) \leq d_{\mid \mathbf{y}}(\mathbf{m})\rbrace} \geq 2\right),
\end{eqnarray}
where the summation runs over all $\mathbf{m'} \in \lbrace 0, 1\rbrace^{k}$.
\end{theorem}
\begin{proof}
The probability of error that $\mathbf{\hat{m}}\neq \mathbf{m}$ can be decomposed into two terms as follows:
\begin{equation}
\label{errordecompose}
\Pr(\mathbf{\hat{m}}\neq \mathbf{m}) = \Pr(\mathbf{\hat{m}}=\varepsilon)+\Pr(\mathbf{\hat{m}}\neq \mathbf{m},\mathbf{\hat{m}} \neq  \varepsilon). 
\end{equation} 
Now, we will give individual upper bound for the RHS terms in \eqref{errordecompose}.

First, observe that any node at branching time $h+1$ is not checked if its parent node $\mathbf{m'}_{1:s(b_{h})}$ satisfies $d_{\mid \mathbf{y}}(\mathbf{m'}_{1:s(b_{h})})> d_{\mid \mathbf{y}}(\mathbf{m})$. The reason is that the stack of the decoding algorithm always contains an ancestor of $\mathbf{m}$ which has a decoding error cost at most as $d_{\mid \mathbf{y}}(\mathbf{m})$ (due to the definition of AEC measures) hence a node with a  decoding error cost lower than of $m'_{1:s(b_{h})}$ always exists. Therefore, the number of nodes that the decoding algorithm checks $N_{c}$ satisfies:
\begin{equation}
\label{upperboundforNc}
N_{c}\leq c_{0}+\displaystyle\sum_{h=1}^{h_{f}-1}c_{h}\displaystyle\sum_{\mathbf{m'}_{1:s(b_{h})}}\mathbb{I}_{\lbrace d_{\mid \mathbf{y}}(\mathbf{m'}_{1:s(b_{h})}) \leq d_{\mid \mathbf{y}}(\mathbf{m})\rbrace}, 
\end{equation}
with probability $1$. Accordingly, we have:
\begin{equation}
\label{Dclebound}
\Pr(\mathbf{\hat{m}}=\varepsilon) \leq D_{\text{CLE}}[s,d].
\end{equation}
Secondly, observe that if $\mathbf{\hat{m}} \neq \varepsilon$, a wrong message is decoded only if there exists at least two terminal nodes $\mathbf{m'}$ that satisfy $d_{\mid \mathbf{y}}(\mathbf{m'}) \leq d_{\mid \mathbf{y}}(\mathbf{m})$ as $\mathbf{m'}=\mathbf{m}$ always satisfies the condition. Accordingly, we have:
\begin{equation}
\label{Dcfebound}
\Pr(\mathbf{\hat{m}}\neq \mathbf{m},\mathbf{\hat{m}} \neq  \varepsilon) \leq D_{\text{CFE}}[s,d].
\end{equation}
Combining \eqref{errordecompose} with \eqref{Dclebound} and \eqref{Dcfebound} proves \eqref{theboundforerrorDcleDcfe}.
\end{proof}
Theorem \ref{DcleDcfeachievability} is an achievability result for a communication system using $(n,k)$-random tree codes with the tree structure $s(t)$ and the SSDGU algorithm with the AEC measure $d(\cdot,\cdot)$ and the computational limit paramater $L$. Let $D_{\text{E}}[s,d]$ represent the RHS in \eqref{theboundforerrorDcleDcfe}, i.e., $D_{\text{E}}[s,d] := D_{\text{CLE}}[s,d]+D_{\text{CFE}}[s,d]$. Then, as the LHS in \eqref{theboundforerrorDcleDcfe} is the sample average of the probability of decoding errors for $(n,k)$-random tree codes with the tree structure $s(t)$, \eqref{theboundforerrorDcleDcfe} guarantees the existence of an $(n,k)$-code that achieves a probability of decoding error  not higher than $D_{\text{E}}[s,d]$ if one uses the SSDGU with  the AEC measure $d(\cdot,\cdot)$ and the computational limit paramater $L$. Particularly, one can find an $(n,k)$-code with a probability of decoding error not higher than $\alpha D_{\text{E}}[s,d]$ for some $\alpha>1$, if sufficient number of   indepedent random tree codes can be tested \footnote{The probability that the probability of decoding error of a sample random tree code being larger than $\alpha D_{\text{E}}[s,d]$ is smaller $\frac{1}{\alpha}$ due to Markov inequality, hence the number of independent tests needed to find such a code is stochastically smaller than a geometrically distributed random variable with parameter $\frac{\alpha-1}{\alpha}$.}.

As \eqref{theboundforerrorDcleDcfe} provides better guarantees on the probability of decoding error with lower $D_{\text{E}}[s,d]$, it is useful to consider the optimization problem of minimizing $D_{\text{E}}[s,d]$:
\begin{equation}
\label{minDE}
\min_{d \in \mathcal{A}_{d,\text{AEC}}}\min_{s \in \mathcal{A}_{s}(n,k)} D_{\text{E}}[s,d],
\end{equation}
where $\mathcal{A}_{d,\text{AEC}}$ is the set of all AEC measures and $\mathcal{A}_{s}(n,k)$ is set of all non-decreasing functions with mapping $\lbrace 1, \cdots, n\rbrace \mapsto \lbrace  1, \cdots,k\rbrace$.     

Let  represent the solution of \eqref{minDE}. A problem that is dual to \eqref{minDE} is the problem of minimizing $L$ while satisfying $D_{\text{E}}[s,d] \leq \epsilon$ for some $\epsilon \in [0,1]$:
\begin{eqnarray}
\label{minL}
&&\min_{d \in \mathcal{A}_{d,\text{AEC}}}\min_{s \in \mathcal{A}_{s}(n,k)} L \\
&& \text{s.t. } D_{\text{E}}[s,d] \leq \epsilon.\nonumber
\end{eqnarray}
Let  $D_{\text{E}}^{*}(n,k,L)$ and $L_{D}(n,k,\epsilon)$ represent the solutions of \eqref{minDE} and \eqref{minL}, respectively. 

Solving \eqref{minDE} is analytically challenging and even if its search space is discretized \footnote{As $\mathcal{A}_{s}(n,k)$ is already discrete and finite, only   $\mathcal{A}_{d,\text{AEC}}$ requires discretization.}, an exhaustive search would require vast \footnote{At least the cardinality of $\mathcal{A}_{s}(n,k)$, i.e, $\vert\mathcal{A}_{s}(n,k)\vert$, which can be expressed as  $\binom{n+k-2}{k-1}$, is factorially large.} computational resources. On the other hand, relaxations of \eqref{minDE} together with heuristic optimization can be useful for finding good $[s,d]$ combinations.

We will consider the functionals that upper bound $D_{\text{E}}[s,d]$ with expressions that are more tractable. One example is the following functional:
\begin{equation}
D_{\text{E}}^{\text{(U)}}[s,d] =  D_{\text{CLE}}^{\text{(U)}}[s,d]+D_{\text{CFE}}^{\text{(U)}}[s,d],
\end{equation}
where
\begin{eqnarray}
\label{DCLEUdefinition}
&&D_{\text{CLE}}^{\text{(U)}}[s,d] = \nonumber\\
&& \mathbb{E} \left[ \left[ \displaystyle\sum_{h=0}^{h_{f}-1}v_{h}\Pr\left(d_{\mid \mathbf{y}}(\mathbf{\bar{m}}_{1:s(b_{h})}) \leq d_{\mid \mathbf{y}}(\mathbf{m})\mid \mathbf{x}, \mathbf{y}\right) \right]_{\leq 1}\right], \nonumber\\
\end{eqnarray}
where $v_{h}=\frac{1}{L}2^{s(b_{h+1})}$, $\mathbf{\bar{m}}_{1:s(b_{h})}$ is  uniformly selected from $\lbrace 0,1\rbrace^{s(b_{h})}$ being independent from other random variables, $d_{\mid \mathbf{y}}(\mathbf{\bar{m}}_{1:s(b_{0})})=0$ w.p.1, and
\begin{eqnarray}
\label{DCFEUdefinition}
D_{\text{CFE}}^{\text{(U)}}[s,d] = \mathbb{E}\left[ \left[ 2^{k}\Pr\left( d_{\mid \mathbf{y}}(\mathbf{\bar{m}}) \leq d_{\mid \mathbf{y}}(\mathbf{m})\mid \mathbf{x}, \mathbf{y}\right) -1\right]_{\leq 1}\right],
\end{eqnarray}
where $\mathbf{\bar{m}}$ is uniformly selected from $\lbrace 0,1\rbrace^{k}$ being independent from other random variables. Then, we have:
\begin{theorem}
The functional $D_{\text{E}}^{\text{(U)}}[s,d]$ is a uniform upper bound for $D_{\text{E}}[s,d]$ in $\mathcal{A}_{d,\text{AEC}}\times \mathcal{A}_{s}(n,k)$, i.e.,
\begin{equation}
\label{DEUbound}
D_{\text{E}}[s,d] \leq D_{\text{E}}^{\text{(U)}}[s,d],
\end{equation} 
for all $s \in \mathcal{A}_{s}(n,k)$ and $d \in \mathcal{A}_{d,\text{AEC}}$.
\end{theorem} 
\begin{proof}
First observe that:
\begin{eqnarray}
\label{mbarforCLE}
&&\displaystyle\sum_{\mathbf{m'}_{1:s(b_{h})}}\mathbb{I}_{\lbrace d_{\mid \mathbf{y}}(\mathbf{m'}_{1:s(b_{h})}) \leq d_{\mid \mathbf{y}}(\mathbf{m})\rbrace}\nonumber\\
&& =2^{s(b_{h})}\mathbb{E}_{\mathbf{\bar{m}}}\left[ \mathbb{I}_{\lbrace d_{\mid \mathbf{y}}(\mathbf{\bar{m}}_{1:s(b_{h})}) \leq d_{\mid \mathbf{y}}(\mathbf{m})\rbrace}\right], \text{w.p.1} 
\end{eqnarray}
where the expectation is taken over only on $\mathbf{\bar{m}}_{1:s(b_{h})}$.
Then, consider:
\begin{eqnarray}
\label{DCLEUbound}
&&D_{\text{CLE}}[s,d] \nonumber\\ 
&&=\Pr\left( c_{0}+\displaystyle\sum_{h=1}^{h_{f}-1}c_{h}\displaystyle\sum_{\mathbf{m'}_{1:s(b_{h})}}\mathbb{I}_{\lbrace d_{\mid \mathbf{y}}(\mathbf{m'}_{1:s(b_{h})}) \leq d_{\mid \mathbf{y}}(\mathbf{m})\rbrace} \geq L\right) \nonumber\\ 
&&=\Pr\left( \frac{c_{0}}{L}+\displaystyle\sum_{h=1}^{h_{f}-1}\frac{c_{h}}{L}\displaystyle\sum_{\mathbf{m'}_{1:s(b_{h})}}\mathbb{I}_{\lbrace d_{\mid \mathbf{y}}(\mathbf{m'}_{1:s(b_{h})}) \leq d_{\mid \mathbf{y}}(\mathbf{m})\rbrace} \geq 1\right) \nonumber\\
&&\leq\mathbb{E} \left[ \left[  \frac{c_{0}}{L}+\displaystyle\sum_{h=1}^{h_{f}-1}\frac{c_{h}}{L}\displaystyle\sum_{\mathbf{m'}_{1:s(b_{h})}}\mathbb{I}_{\lbrace d_{\mid \mathbf{y}}(\mathbf{m'}_{1:s(b_{h})}) \leq d_{\mid \mathbf{y}}(\mathbf{m})\rbrace}\right]_{\leq 1}\right]  \nonumber\\
&&=\mathbb{E} \left[ \left[ \displaystyle\sum_{h=0}^{h_{f}-1}\frac{c_{h}2^{s(b_{h})}}{L}\mathbb{E}_{\mathbf{\bar{m}}}\left[ \mathbb{I}_{\lbrace d_{\mid \mathbf{y}}(\mathbf{\bar{m}}_{1:s(b_{h})}) \leq d_{\mid \mathbf{y}}(\mathbf{m})\rbrace}\right]\right]_{\leq 1}\right]  \nonumber\\
&&=\mathbb{E} \left[ \left[ \displaystyle\sum_{h=0}^{h_{f}-1}\frac{2^{s(b_{h+1})}}{L}\mathbb{E}_{\mathbf{\bar{m}}}\left[ \mathbb{I}_{\lbrace d_{\mid \mathbf{y}}(\mathbf{\bar{m}}_{1:s(b_{h})}) \leq d_{\mid \mathbf{y}}(\mathbf{m})\rbrace}\right]\right]_{\leq 1}\right]  \nonumber\\
&&\leq\mathbb{E} \left[ \left[ \displaystyle\sum_{h=0}^{h_{f}-1}v_{h}\Pr\left(d_{\mid \mathbf{y}}(\mathbf{\bar{m}}_{1:s(b_{h})}) \leq d_{\mid \mathbf{y}}(\mathbf{m})\mid \mathbf{x}, \mathbf{y}\right) \right]_{\leq 1}\right]\nonumber\\
&&=D_{\text{CLE}}^{\text{(U)}}[s,d],  
\end{eqnarray}
where the first inequality is due to that $\Pr(X \geq 1)\leq \mathbb{E}[[X]_{\leq 1}]$ for a non-negative random variable $X$, the third equality is due to \eqref{mbarforCLE}, the fourth equality is due to $c_{h}=2^{s(b_{h+1})-s(b_{h})}$, and the second inequality is due to Jensen's inequality for the commutation of $[\cdot]_{\leq 1}$ and $\mathbb{E}[\cdot \mid \mathbf{x}, \mathbf{y}, \mathbf{\bar{m}}]$ operations. 

Similar to \eqref{mbarforCLE}, observe that:
\begin{eqnarray}
\label{mbarforCFE}
\displaystyle\sum_{\mathbf{m'}}\mathbb{I}_{\lbrace d_{\mid \mathbf{y}}(\mathbf{m'}) \leq d_{\mid \mathbf{y}}(\mathbf{m})\rbrace}=2^{k}\mathbb{E}_{\mathbf{\bar{m}}}\left[ \mathbb{I}_{\lbrace d_{\mid \mathbf{y}}(\mathbf{\bar{m}}) \leq d_{\mid \mathbf{y}}(\mathbf{m})\rbrace}\right], \text{w.p.1} \nonumber\\
\end{eqnarray}
where the expectation is taken over only on  $\mathbf{\bar{m}}$.
Then, consider:
\begin{eqnarray}
\label{DCFEUbound}
&&D_{\text{CFE}}[s,d] \nonumber\\
&&=\Pr\left( \displaystyle\sum_{\mathbf{m'}}\mathbb{I}_{\lbrace d_{\mid \mathbf{y}}(\mathbf{m'}) \leq d_{\mid \mathbf{y}}(\mathbf{m})\rbrace} \geq 2\right)\nonumber\\
&&=\Pr\left( \displaystyle\sum_{\mathbf{m'}}\mathbb{I}_{\lbrace d_{\mid \mathbf{y}}(\mathbf{m'}) \leq d_{\mid \mathbf{y}}(\mathbf{m})\rbrace}-1 \geq 1\right)\nonumber\\
&&=\mathbb{E}\left[ \left[ \displaystyle\sum_{\mathbf{m'}}\mathbb{I}_{\lbrace d_{\mid \mathbf{y}}(\mathbf{m'}) \leq d_{\mid \mathbf{y}}(\mathbf{m})\rbrace}-1\right]_{\leq 1}\right] \nonumber\\
&&=\mathbb{E}\left[ \left[ 2^{k}\mathbb{E}_{\mathbf{\bar{m}}}\left[ \mathbb{I}_{\lbrace d_{\mid \mathbf{y}}(\mathbf{\bar{m}}) \leq d_{\mid \mathbf{y}}(\mathbf{m})}\right]-1\right]_{\leq 1}\right] \nonumber\\
&&\leq\mathbb{E}\left[ \left[ 2^{k}\Pr\left( d_{\mid \mathbf{y}}(\mathbf{\bar{m}}) \leq d_{\mid \mathbf{y}}(\mathbf{m})\mid \mathbf{x}, \mathbf{y}\right) -1\right]_{\leq 1}\right] \nonumber\\
&&=D_{\text{CFE}}^{(U)}[s,d],
\end{eqnarray}
where the third equality is due to that $\Pr(X \geq 1)=\mathbb{E}[[X]_{\leq 1}]$ for $X \in \lbrace 0, 1, \cdots, 2^{k}\rbrace$, the fourth equality is due to \eqref{mbarforCFE}, and the inequality is due to Jensen's inequality for the commutation of $[\cdot]_{\leq 1}$ and $\mathbb{E}[\cdot \mid \mathbf{x}, \mathbf{y}, \mathbf{\bar{m}}]$ operations.

Combining \eqref{DCLEUbound} and \eqref{DCFEUbound} proves \eqref{DEUbound}.
\end{proof}
An upper bound on $D_{\text{CLE}}^{\text{(U)}}[s,d]$, hence a weaker upper bound on $D_{\text{CLE}}[s,d]$, is obtained if $[\cdot]_{\leq 1}$ operation is omitted in \eqref{DCLEUdefinition}:
\begin{eqnarray}
&&D_{\text{CLE}}^{\text{(M)}}[s,d] = \nonumber\\
&& \mathbb{E} \left[ \displaystyle\sum_{h=0}^{h_{f}-1}v_{h}\Pr\left(d_{\mid \mathbf{y}}(\mathbf{\bar{m}}_{1:s(b_{h})}) \leq d_{\mid \mathbf{y}}(\mathbf{m})\mid \mathbf{x}, \mathbf{y}\right) \right], \nonumber\\
\end{eqnarray}
is an upper bound on $D_{\text{CLE}}^{\text{(U)}}[s,d]$ as $[X]_{\leq 1} \leq X$ w.p.1 for any random real variable $X$. The bound that $D_{\text{CLE}}[s,d] \leq D_{\text{CLE}}^{\text{(M)}}[s,d]$ could be also obtained appyling Markov's inequality on $D_{\text{CLE}}^{\text{(M)}}[s,d]$ and this relation makes the bound more appealing as we have the following:
\begin{theorem}
\label{expectedNcboundDCLEML}
If $N_{c}$ is the number of nodes that the decoding algorithm checks for a random run of the SSDGU algorithm, then
\begin{equation}
\label{upperboundforexpectedNc}
\mathbb{E}[N_{c}] \leq  D_{\text{CLE}}^{\text{(M)}}[s,d]L,
\end{equation}  
for all $s \in \mathcal{A}_{s}(n,k)$ and $d \in \mathcal{A}_{d,\text{AEC}}$.
\end{theorem}
\begin{proof}
Applying expectation for both sides on \eqref{upperboundforNc} gives \eqref{upperboundforexpectedNc}.
\end{proof}
Theorem \ref{expectedNcboundDCLEML} shows that the expected number of node checks, i.e., $\mathbb{E}[N_{c}]$, is a small fraction of $L$ given that $D_{\text{CLE}}^{\text{(M)}}[s,d]$ is small. Hence, $D_{\text{CLE}}^{\text{(M)}}[s,d]$ is relevant considering the average run time of the SSDGU algorithm as it is proportial to $\mathbb{E}[N_{c}]$.

Another advantage of the relaxation of $D_{\text{CLE}}^{\text{(U)}}[s,d]$ to $D_{\text{CLE}}^{\text{(M)}}[s,d]$ is that the latter can be expressed in the following form which is both analytically and computationally more convenient: 
\begin{eqnarray}
\label{DCLEMconvenient}
D_{\text{CLE}}^{\text{(M)}}[s,d] =  \displaystyle\sum_{h=0}^{h_{f}-1}v_{h} \Pr\left(d_{\mid \mathbf{y}}(\mathbf{\bar{m}}_{1:s(b_{h})}) \leq d_{\mid \mathbf{y}}(\mathbf{m})\right).
\end{eqnarray} 
While $D_{\text{CLE}}[s,d]$ and its relaxations, i.e., $D_{\text{CLE}}^{\text{(U)}}[s,d]$ and $D_{\text{CLE}}^{\text{(M)}}[s,d]$, depend on computational constraints through the computational limit $L$, $D_{\text{CFE}}[s,d]$ and its relaxation $D_{\text{CFE}}^{\text{(U)}}[s,d]$ are free from computational constraints and characterize the probability of decoding error without computational constraints. In fact, the latter can be related to known random coding achievability bounds as in the following:
\begin{theorem}
\label{RCUconnectionforDCFEU}
If $\mathcal{Y}$ is discrete, $s(1)=k$ and
\begin{equation}
\label{loglikelihoodford}
d(\mathbf{x'}_{1:t},\mathbf{y'}_{1:t})= -\log_{2}\left( P_{\mathbf{y}_{1:t} \mid \mathbf{x}_{1:t}}\left( \mathbf{y'}_{1:t} \mid \mathbf{x'}_{1:t}\right) \right), 
\end{equation}
then
\begin{equation}
\label{RCUcaseforDCFEU}
D_{\text{CFE}}^{\text{(U)}}[s,d] = \mathbb{E}\left[ \left[ \left( 2^{k}-1\right)\Pr\left( i(\mathbf{\bar{x}},\mathbf{y}) \geq i(\mathbf{x},\mathbf{y})\mid \mathbf{x}, \mathbf{y}\right)\right]_{\leq 1}\right],
\end{equation}
where $i(\mathbf{x'},\mathbf{y'})=\log_{2}\left( \frac{P_{\mathbf{y} \mid \mathbf{x}}\left( \mathbf{y'} \mid \mathbf{x'}\right) }{P_{\mathbf{y}}\left( \mathbf{y'}\right) }\right)$ is the information density and $\mathbf{\bar{x}}$ is uniformly selected from $\lbrace 0,1\rbrace^{n}$ being independent from other random variables.
\end{theorem}
\begin{proof}
If $s(1)=k$, then $t=1$ is the only branching time and hence $\mathbf{G}\mathbf{\bar{m}}$ and $\mathbf{G}\mathbf{m}$ are conditionally indepedent for the conditioning on $\mathbf{\bar{m}} \neq \mathbf{m}$. Accordingly,
\begin{eqnarray}
\label{condprobforindependentx}
&&\Pr\left( d_{\mid \mathbf{y}}(\mathbf{\bar{m}}) \leq d_{\mid \mathbf{y}}(\mathbf{m})\mid \mathbf{x}, \mathbf{y}\right) \nonumber\\
&&=\Pr\left( d_{\mid \mathbf{y}}(\mathbf{\bar{m}}) \leq d_{\mid \mathbf{y}}(\mathbf{m})\mid \mathbf{x}, \mathbf{y}, \mathbf{\bar{m}} \neq \mathbf{m}\right)\Pr(\mathbf{\bar{m}} \neq \mathbf{m})
\nonumber\\
&&+\Pr(\mathbf{\bar{m}}=\mathbf{m})
\nonumber\\
&&=\Pr\left( d(\mathbf{\bar{x}}, \mathbf{y} ) \leq d(\mathbf{x}, \mathbf{y} )\mid \mathbf{x}, \mathbf{y}, \mathbf{\bar{m}} \neq \mathbf{m}\right)\Pr(\mathbf{\bar{m}} \neq \mathbf{m})
\nonumber\\
&&+\Pr(\mathbf{\bar{m}}=\mathbf{m})
\nonumber\\
&&=\Pr\left( d(\mathbf{\bar{x}}, \mathbf{y} ) \leq d(\mathbf{x}, \mathbf{y} )\mid \mathbf{x}, \mathbf{y}\right)\Pr(\mathbf{\bar{m}} \neq \mathbf{m})
\nonumber\\
&&+\Pr(\mathbf{\bar{m}}=\mathbf{m})
\nonumber\\
&&=\Pr\left( d(\mathbf{\bar{x}}, \mathbf{y} ) \leq d(\mathbf{x}, \mathbf{y} )\mid \mathbf{x}, \mathbf{y}\right)\left(1-2^{-k}\right) +2^{-k}, \text{w.p.1}, \nonumber\\
\end{eqnarray}
where the second equality is due that $\mathbf{G}\mathbf{\bar{m}}$ and $\mathbf{G}\mathbf{m}$ are conditionally indepedent for the conditioning on $\mathbf{\bar{m}} \neq \mathbf{m}$.

Combining \eqref{condprobforindependentx} and \eqref{loglikelihoodford} with \eqref{DCFEUdefinition} gives \eqref{RCUcaseforDCFEU}.
\end{proof}
Theorem \ref{RCUconnectionforDCFEU} shows that a special case of $D_{\text{CFE}}^{\text{(U)}}[s,d]$ is the random coding union (RCU) bound for uniform input distributions.    In other words, $D_{\text{CFE}}[s,d]$ in \eqref{theboundforerrorDcleDcfe} can be as small as the RCU bound for uniform input distributions with the particular choice of  $s(1)=k$, i.e., $s(t)=k, \forall t$, and the decoding measure in \eqref{loglikelihoodford} given that it is in $\mathcal{A}_{d,\text{AEC}}$. On the other hand, the decoding measure in \eqref{loglikelihoodford} is not necessarly included in $\mathcal{A}_{d,\text{AEC}}$. Yet, we can show that the decoding measure in \eqref{loglikelihoodford} is in $\mathcal{A}_{d,\text{AEC}}$ for channels that are \emph{causal} in the following sense:
\begin{definition}
A channel $P_{\mathbf{y} \mid \mathbf{x}}$ is said to be causal if it satisfies:
\begin{equation}
P_{\mathbf{y}_{1:t} \mid \mathbf{x}_{1:t}}(\mathbf{y'}_{1:t} \mid \mathbf{x'}_{1:t})=\displaystyle\prod_{t'=1}^{t}P_{\mathbf{y}_{t'} \mid \mathbf{x}_{1:t'}}(\mathbf{y'}_{t'} \mid \mathbf{x'}_{1:t'}),
\end{equation}
for all $t \in \lbrace 1,\cdots, n\rbrace$.
\end{definition} 
Accordingly, we have 
\begin{theorem}
If the channel $P_{\mathbf{y} \mid \mathbf{x}}$ is causal, then the decoding measure in \eqref{loglikelihoodford} is in $\mathcal{A}_{d,\text{AEC}}$ .
\end{theorem}
\begin{proof}
Consider a node $\mathbf{m'}_{1:s(t_{1})}$ and its descendant $\mathbf{m'}_{1:s(t_{2})}$ where $1 \leq s(t_{1}) < s(t_{2}) \leq k$ and $1 \leq t_{1} < t_{2} \leq n$ where $t_{1}=r_{s(1)}$ and $t_{2}=r_{s(2)}$. Observe that:
\begin{eqnarray}
&& d(\varepsilon_{t_{2}}(\mathbf{m'}_{1:s(t_{2})}),\mathbf{y}_{1:t_{2}}) \nonumber\\
&& = -\log_{2}\left( P_{\mathbf{y}_{1:t_{2}} \mid \mathbf{x}_{1:t_{2}}}(\mathbf{y}_{1:t_{2}} \mid \varepsilon_{t_{2}}(\mathbf{m'}_{1:s(t_{2})}))\right) \nonumber\\
&& = -\log_{2}\left( P_{\mathbf{y}_{1:t_{2}} \mid \mathbf{x}_{1:t_{2}}}(\mathbf{y}_{1:t_{2}} \mid \mathbf{x'}_{1:t_{2}})\right) 
\nonumber\\
&& = -\log_{2}\left( \displaystyle\prod_{t'=1}^{t_{2}}P_{\mathbf{y}_{t'} \mid \mathbf{x}_{1:t'}}(\mathbf{y'}_{t'} \mid \mathbf{x'}_{1:t'})\right) \nonumber\\
&& = -\log_{2}\left( \displaystyle\prod_{t'=1}^{t_{1}}P_{\mathbf{y}_{t'} \mid \mathbf{x}_{1:t'}}(\mathbf{y'}_{t'} \mid \mathbf{x'}_{1:t'})\right)- \nonumber\\
&&\log_{2}\left( \displaystyle\prod_{t'=t_{1}}^{t_{2}}P_{\mathbf{y}_{t'} \mid \mathbf{x}_{1:t'}}(\mathbf{y'}_{t'} \mid \mathbf{x'}_{1:t'})\right)\nonumber\\
&& = d(\varepsilon_{t_{1}}(\mathbf{m'}_{1:s(t_{1})}),\mathbf{y}_{1:t_{1}})- \nonumber\\
&&\log_{2}\left( \displaystyle\prod_{t'=t_{1}}^{t_{2}}P_{\mathbf{y}_{t'} \mid \mathbf{x}_{1:t'}}(\mathbf{y'}_{t'} \mid \mathbf{x'}_{1:t'})\right)\nonumber\\
&& \geq d(\varepsilon_{t_{1}}(\mathbf{m'}_{1:s(t_{1})}),\mathbf{y}_{1:t_{1}}),
\end{eqnarray}
where the third and fifth inequalities follow from the definition of channels that are causal, and the inequality follows from the fact that $-\log_{2}(P) \geq 0$ for $P \in [0,1]$.
\end{proof}
We will consider memoryless channels in particular as they are the simplest class of channels that are causal. For memoryless channels, the following decomposition of the decoding measure $d$ will be useful,
\begin{equation}
d(\mathbf{x'}_{1:t},\mathbf{y'}_{1:t})= \displaystyle\sum_{t'=1}^{t}d_{t'}(\mathbf{x'}_{t'},\mathbf{y'}_{t'}),
\end{equation} 
where $d_{t'}(\mathbf{x'}_{t'},\mathbf{y'}_{t'})$ terms are independent if either $\mathbf{x'}$ is equal to $\mathbf{x}$, i.e., the input of the channel, or $\mathbf{\bar{x}}$  , i.e., an independently selected binary sequence in $\lbrace 0,1\rbrace^{n}$, while $\mathbf{y'}= \mathbf{y}$, i.e., the output of the channel, as the channel is memorlyless.

Accordingly, we can decompose $\Pr\left(d_{\mid \mathbf{y}}(\mathbf{\bar{m}}_{1:s(b_{h})}) \leq d_{\mid \mathbf{y}}(\mathbf{m})\right)$ in \eqref{DCLEMconvenient} as follows:
\begin{eqnarray}
&&\Pr\left(d_{\mid \mathbf{y}}(\mathbf{\bar{m}}_{1:s(b_{h})}) \leq d_{\mid \mathbf{y}}(\mathbf{m})\right) \nonumber\\
&&=\displaystyle\sum_{h'=0}^{h} \Pr\left(d_{\mid \mathbf{y}}(\mathbf{\bar{m}}_{1:s(b_{h})}) \leq d_{\mid \mathbf{y}}(\mathbf{m}) \mid \tau = b_{h'}\right)\times
\nonumber\\
&&\Pr(\tau_{h} = b_{h'})
\nonumber\\
&&=\displaystyle\sum_{h'=0}^{h} \Pr\left(d_{b_{h'+1}:r_{[h]}}(\mathbf{\bar{x}},\mathbf{y}) \leq d_{b_{h'+1}:n}(\mathbf{x},\mathbf{y})\right)\times
\nonumber\\
&&\Pr(\tau_{h} = b_{h'}),
\end{eqnarray}
where $r_{[h]}=r_{s(b_{h})}$ and
\begin{equation}
d_{t_{s}:t_{f}}(\mathbf{x'},\mathbf{y'})=\displaystyle\sum_{t=t_{s}}^{t_{f}}d_{t}(\mathbf{x'}_{t},\mathbf{y'}_{t}),
\end{equation}
where $d_{t_{s}:t_{f}}(\mathbf{x'},\mathbf{y'})=0$ if $t_{s} > t_{f}$,
and
\begin{equation}
\tau_{h} = \sup \lbrace \tau' \in \lbrace b_{1},b_{2}, \cdots, b_{h}\rbrace : \mathbf{\bar{m}}_{1:s(\tau')} = \mathbf{m}_{1:s(\tau')} \rbrace,
\end{equation}
if $\mathbf{\bar{m}}_{1:s(b_{1})} = \mathbf{m}_{1:s(b_{1})}$, otherwise $\tau_{h}=b_{0}$ for some $b_{0}<b_{1}$.
Similarly, we have:
\begin{eqnarray}
&&\Pr\left( d_{\mid \mathbf{y}}(\mathbf{\bar{m}}) \leq d_{\mid \mathbf{y}}(\mathbf{m})\mid \mathbf{x}, \mathbf{y},\tau = b_{h}\right)\nonumber\\
&&=\Pr\left( d_{t_{h}: n}(\mathbf{\bar{x}},\mathbf{y}) \leq d_{t_{h} :n}(\mathbf{x},\mathbf{y})\mid \mathbf{x}, \mathbf{y},\tau = b_{h}\right)\nonumber\\
&&=\Pr\left( d_{t_{h}: n}(\mathbf{\bar{x}},\mathbf{y}) \leq d_{t_{h} :n}(\mathbf{x},\mathbf{y})\mid \mathbf{x}, \mathbf{y}\right), \text{w.p.1} \nonumber\\
\end{eqnarray} 
where $\tau= \tau_{h_{f}}$ and $t_{h}=b_{h+1}$ if $h < h_{f}$,   and $t_{h}>n$ if $h=h_{f}$.

Note that
\begin{equation}
\Pr(\tau_{h} = b_{h'})= 2^{-s(b_{h'})}-2^{-s(b_{h'+1})},
\end{equation}
for $h> h'>0$, $\Pr(\tau_{h} = b_{0})=1-2^{-s(b_{1})}$,
and $\Pr(\tau_{h} = b_{h})=2^{-s(b_{h})}$,
and
\begin{equation}
\Pr(\tau = b_{h} \mid \mathbf{x}, \mathbf{y})= \Pr(\tau = b_{h}), \text{w.p.1}
\end{equation}
where $\Pr(\tau = b_{h})=2^{-s(b_{h})}-2^{-s(b_{h+1})}$ for $0 < h < h_{f}$, $\Pr(\tau = b_{0})=1-2^{-s(b_{1})}$, and $\Pr(\tau = b_{h_{f}})=2^{-k}$.

Therefore, we can express $D_{\text{CLE}}^{\text{(M)}}[s,d]$ as:
\begin{eqnarray}
\label{memorylessDCLE}
&&D_{\text{CLE}}^{\text{(M)}}[s,d]=\nonumber\\
&&\displaystyle\sum_{h=0}^{h_{f}-1}\displaystyle\sum_{h'=0}^{h} v_{h;h'} \Pr\left(d_{b_{h'+1}:r_{[h]}}(\mathbf{\bar{x}},\mathbf{y}) \leq d_{b_{h'+1}:n}(\mathbf{x},\mathbf{y})\right),\nonumber\\
\end{eqnarray}
where $v_{h;h'}=\frac{1}{L}c_{h}\Pr(\tau_{h} = b_{h'})$, and
\begin{eqnarray}
\label{memorylessDCFE}
&&D_{\text{CFE}}^{\text{(U)}}[s,d]=\nonumber\\ 
&&\mathbb{E}\left[ \left[\displaystyle\sum_{h=0}^{h_{f}-1}w_{h}\Pr\left( d_{t_{h}: n}(\mathbf{\bar{x}},\mathbf{y}) \leq d_{t_{h} :n}(\mathbf{x},\mathbf{y})\mid \mathbf{x}, \mathbf{y}\right)\right]_{\leq 1}\right],\nonumber\\ 
\end{eqnarray}
where $w_{h}=2^{k}\Pr(\tau = b_{h})$.

The techniques that relax RCU bound to Gallager's bound can be used to obtain loose bounds from \eqref{memorylessDCLE} and \eqref{memorylessDCFE} for memoryless channels as in the following result.

\begin{theorem}
If $P_{\mathbf{y} \mid \mathbf{x}}$ is a memoryless channel, then the functional $D_{\text{E}}^{\text{(C)}}[s,d]$ is a uniform upper bound for $D_{\text{E}}^{\text{(U)}}[s,d]$ in $\mathcal{A}_{d,\text{AEC}}\times \mathcal{A}_{s}(n,k)$, i.e.,
\begin{equation}
\label{DECbound}
D_{\text{E}}[s,d]^{\text{(U)}} \leq D_{\text{E}}^{\text{(C)}}[s,d],
\end{equation} 
for all $s \in \mathcal{A}_{s}(n,k)$ and $d \in \mathcal{A}_{d,\text{AEC}}$,
where
\begin{equation}
D_{\text{E}}^{\text{(C)}}[s,d]= D_{\text{CLE}}^{\text{(C)}}[s,d]+D_{\text{CFE}}^{\text{(C)}}[s,d],
\end{equation}
where
\begin{eqnarray}
&&D_{\text{CLE}}^{\text{(C)}}[s,d] = \nonumber\\
&&\displaystyle\sum_{h=0}^{h_{f}-1}\displaystyle\sum_{h'=0}^{h} v_{h;h'} \mathbb{E}\left[ 2^{\vartheta_{h;h'}\left[ d_{b_{h'+1}:n}(\mathbf{x},\mathbf{y})-d_{b_{h'+1}:r_{[h]}}(\mathbf{\bar{x}},\mathbf{y})\right]  }\right]^{\varrho_{h;h'}},\nonumber\\
\end{eqnarray}
for $\vartheta_{h;h'} \geq 0$ and $\varrho_{h;h'}\in [0,1]$, and
\begin{eqnarray}
&&D_{\text{CFE}}^{\text{(C)}}[s,d] = \nonumber\\
&&\mathbb{E}\left[ \displaystyle\sum_{h=0}^{h_{f}-1}w_{h}^{\rho_{h}}\mathbb{E}\left[  2^{\theta_{h}\left[ d_{t_{h} :n}(\mathbf{x},\mathbf{y})-d_{t_{h}: n}(\mathbf{\bar{x}},\mathbf{y}) \right]}\mid \mathbf{y}\right]^{\rho_{h}}\right],\nonumber\\
\end{eqnarray}
for $\theta_{h} \geq 0$ and $\rho_{h} \in [0,1]$.
\end{theorem}
\begin{proof}
Observe that:
\begin{eqnarray}
&&D_{\text{CLE}}^{\text{(U)}}[s,d] \leq \nonumber\\
&&\displaystyle\sum_{h=0}^{h_{f}-1}\displaystyle\sum_{h'=0}^{h} v_{h;h'} \Pr\left(d_{b_{h'+1}:r_{[h]}}(\mathbf{\bar{x}},\mathbf{y}) \leq d_{b_{h'+1}:n}(\mathbf{x},\mathbf{y})\right)\nonumber\\
&&\displaystyle\sum_{h=0}^{h_{f}-1}\displaystyle\sum_{h'=0}^{h} v_{h;h'} \Pr\left(d_{b_{h'+1}:r_{[h]}}(\mathbf{\bar{x}},\mathbf{y}) \leq d_{b_{h'+1}:n}(\mathbf{x},\mathbf{y})\right)^{\varrho_{h;h'}}\nonumber\\
&&\displaystyle\sum_{h=0}^{h_{f}-1}\displaystyle\sum_{h'=0}^{h} v_{h;h'} \mathbb{E}\left[ 2^{\vartheta_{h;h'}\left[ d_{b_{h'+1}:n}(\mathbf{x},\mathbf{y})-d_{b_{h'+1}:r_{[h]}}(\mathbf{\bar{x}},\mathbf{y})\right]  }\right]^{\varrho_{h;h'}}\nonumber\\
&&=D_{\text{CLE}}^{\text{(C)}}[s,d], 
\end{eqnarray}
where the first inequality follows from $D_{\text{CLE}}^{\text{(U)}}[s,d]\leq D_{\text{CLE}}^{\text{(M)}}[s,d]$ and \eqref{memorylessDCLE}, the second inequality follows from that $a \leq a^{\varrho}$ for $a, \varrho \in [0,1]$, and the third inequality follows from Chernoff bound.

Then, observe that:
\begin{eqnarray}
&&D_{\text{CFE}}^{\text{(U)}}[s,d]\nonumber\\
&&\leq\mathbb{E}\left[ \left[\displaystyle\sum_{h=0}^{h_{f}-1}w_{h}\Pr\left( d_{t_{h}: n}(\mathbf{\bar{x}},\mathbf{y}) \leq d_{t_{h} :n}(\mathbf{x},\mathbf{y})\mid \mathbf{y}\right)\right]_{\leq 1}\right]\nonumber\\
&&\leq\mathbb{E}\left[ \displaystyle\sum_{h=0}^{h_{f}-1}\left[w_{h}\Pr\left( d_{t_{h}: n}(\mathbf{\bar{x}},\mathbf{y}) \leq d_{t_{h} :n}(\mathbf{x},\mathbf{y})\mid \mathbf{y}\right)\right]_{\leq 1}\right]\nonumber\\
&&\leq\mathbb{E}\left[ \displaystyle\sum_{h=0}^{h_{f}-1}\left[w_{h}\Pr\left( d_{t_{h}: n}(\mathbf{\bar{x}},\mathbf{y}) \leq d_{t_{h} :n}(\mathbf{x},\mathbf{y})\mid \mathbf{y}\right)\right]^{\rho_{h}}\right]\nonumber\\
&&\leq\mathbb{E}\left[ \displaystyle\sum_{h=0}^{h_{f}-1}w_{h}^{\rho_{h}}\mathbb{E}\left[  2^{\theta_{h}\left[ d_{t_{h} :n}(\mathbf{x},\mathbf{y})-d_{t_{h}: n}(\mathbf{\bar{x}},\mathbf{y}) \right]}\mid \mathbf{y}\right]^{\rho_{h}}\right] \nonumber\\
&& = D_{\text{CFE}}^{\text{(C)}}[s,d],
\end{eqnarray}
where the first inequality follows from applying Jensen's inequality on \eqref{memorylessDCFE} for the commutation of $[\cdot]_{\leq 1}$ and $\mathbb{E}[\cdot \mid \mathbf{y}]$ operations, the second inequlity follows from that $[a + b]_{\leq 1}\leq [a]_{\leq 1}+[b]_{\leq 1}$, the third inequlity follows from that $[a]_{\leq 1} \leq a^{\rho}$ for $\rho \in [0,1]$ and the fourth inequlity follows from Chernoff bound.
\end{proof}
Note that the functional $D_{\text{E}}^{\text{(C)}}[s,d]$ is easier to evaluate as we have:
\begin{eqnarray}
\label{varthetacosts}
&&\mathbb{E}\left[ 2^{\vartheta_{h;h'}\left[ d_{b_{h'+1}:n}(\mathbf{x},\mathbf{y})-d_{b_{h'+1}:r_{[h]}}(\mathbf{\bar{x}},\mathbf{y})\right]  }\right]=\nonumber\\
&&\displaystyle\prod_{t=b_{h'+1}}^{r_{[h]}}\mathbb{E}\left[ 2^{\vartheta_{h;h'}\left[ d_{t}(\mathbf{x}_{t},\mathbf{y}_{t})-d_{t}(\mathbf{\bar{x}}_{t},\mathbf{y}_{t})\right]}\right]\times\nonumber\\
&&\displaystyle\prod_{t=b_{h}+1}^{n}\mathbb{E}\left[ 2^{\vartheta_{h;h'}d_{t}(\mathbf{x}_{t},\mathbf{y}_{t})}\right],
\end{eqnarray}
and similarly,
\begin{eqnarray}
\label{thetacosts}
&&\mathbb{E}\left[  2^{\theta_{h}\left[ d_{t_{h} :n}(\mathbf{x},\mathbf{y})-d_{t_{h}: n}(\mathbf{\bar{x}},\mathbf{y}) \right]}\mid \mathbf{y}\right]=\nonumber\\
&&\displaystyle\prod_{t=t_{h}}^{n}\mathbb{E}\left[ 2^{\theta_{h}\left[ d_{t}(\mathbf{x}_{t},\mathbf{y}_{t})-d_{t}(\mathbf{\bar{x}}_{t},\mathbf{y}_{t})\right]}\mid \mathbf{y}_{t}\right], \text{w.p.1},
\end{eqnarray}
due to the memorylessness of the channel.

Note that the optimization of $D_{\text{E}}^{\text{(U)}}[s,d]$ over $\vartheta_{h;h'} \geq 0$ and $\varrho_{h;h'} \in [0,1]$ parameters can be done independently from $s(\cdot)$ as the optimization reduces to optimizing the expectation in \eqref{varthetacosts} for every $\vartheta_{h;h'}$ and $\varrho_{h;h'}$ while the expectation is free from $s(\cdot)$. Hence, we can use the precomputed optimal values of the expectation in \eqref{varthetacosts} as we evaluate $D_{\text{E}}^{\text{(U)}}[s,d]$ for distinct $s(\cdot)$ where $d(\cdot,\cdot)$ is fixed. 

In order to be compatible with Gallager's bound and computational ease, we will consider a special case of $D_{\text{E}}^{\text{(C)}}[s,d]$ where $\vartheta_{h;h'}=\frac{1}{1+\varrho_{h;h'}}$, $\theta_{h}=\frac{1}{1+\rho_{h}}$, $\varrho_{h;h'}= \varrho$ and  $\rho_{h}=\rho$ for $\varrho$ and $\rho$ are optimized over $[0,1]$ to minimize the bound, which we will refer to as $D_{\text{E}}^{\text{(G)}}[s,d]$. 

For simplifying the optimization over $d(\cdot,\cdot)$, we will consider a parametrized family of $d(\cdot,\cdot)$ such that:
\begin{equation}
\label{dtcostforgammafromt}
d_{t+1}(\mathbf{x'}_{t+1},\mathbf{y'}_{t+1})=\gamma d_{t}(\mathbf{x'}_{t+1},\mathbf{y'}_{t+1}),
\end{equation} 
or equivalently,
\begin{equation}
\label{dtcostforgammafromthestart}
d_{t}(\mathbf{x'}_{t},\mathbf{y'}_{t})=\gamma^{t-1} d_{1}(\mathbf{x'}_{t},\mathbf{y'}_{t}),
\end{equation}
where $\gamma \in (0,1]$ is the discount factor for the decoding error cost.

We suggest the decoding error cost measures of the form \eqref{dtcostforgammafromthestart} with the motivation from following observations: (1) In \eqref{varthetacosts}, the expectation grows due to the product terms for $t \in [b_{h}+1,n]$ which are always greater than equal to $1$. (2) If $\mathbf{\bar{m}}_{1:s(b_{h})}$ diverges from $\mathbf{m}$ at an early time, i.e., $\tau_{h}$ is close to $1$, then the inputs corresponding to $\mathbf{\bar{m}}_{1:s(b_{h})}$, i.e., $\mathbf{\bar{x}}_{1:b_{h}}$, are likely to have more disagreement with the output $\mathbf{y}$ at early times compared to cases where $\mathbf{\bar{m}}_{1:s(b_{h})}$ diverges from $\mathbf{m}$ at later times, i.e., $\tau_{h}$ is close to $n$.  Accordingly, it is useful to degrade decoding error costs over time. Also note that, the CFE part of the bound does not benefit from the degrading of the decoding error costs or having $\gamma < 1$ as the CFE part increases with decreasing $\gamma$.

We will consider the following decoding error cost measure for a binary symmetric channel (BSC) with crossover probability $p<\frac{1}{2}$:
\begin{equation}
\label{decodingcostwithdiscounts}
d_{t}(\mathbf{x'}_{t},\mathbf{y'}_{t})=\gamma^{t-1} \log_{2}\left( \frac{1-p}{p}\right) \mathbb{I}_{\lbrace \mathbf{x'}_{t} \neq	\mathbf{y'}_{t}\rbrace}.
\end{equation}

For BSC with uniformly distributed inputs, the events $\mathbf{x}_{t} \neq	\mathbf{y}_{t}$ and $\mathbf{\bar{x}}_{t}\neq	\mathbf{y}_{t}$ are indepedent. Accordingly, we have:
\begin{eqnarray}
\label{varthetacostswithinpedentmisses}
&&\mathbb{E}\left[ 2^{\vartheta_{h;h'}\left[ d_{b_{h'+1}:n}(\mathbf{x},\mathbf{y})-d_{b_{h'+1}:r_{[h]}}(\mathbf{\bar{x}},\mathbf{y})\right]  }\right]=\nonumber\\
&&\displaystyle\prod_{t=b_{h'+1}}^{r_{[h]}}\mathbb{E}\left[ 2^{-\vartheta_{h;h'}d_{t}(\mathbf{\bar{x}}_{t},\mathbf{y}_{t})}\right]\displaystyle\prod_{t=b_{h'+1}}^{n}\mathbb{E}\left[ 2^{\vartheta_{h;h'}d_{t}(\mathbf{x}_{t},\mathbf{y}_{t})}\right],\nonumber\\
\end{eqnarray}
and
\begin{eqnarray}
\label{thetacostswithindependentmisses}
&&\mathbb{E}\left[ 2^{\theta_{h}\left[ d_{t_{h} :n}(\mathbf{x},\mathbf{y})-d_{t_{h}: n}(\mathbf{\bar{x}},\mathbf{y}) \right]}\mid \mathbf{y}\right]=\nonumber\\
&&\displaystyle\prod_{t=t_{h}}^{n}\mathbb{E}\left[ 2^{-\theta_{h}d_{t}(\mathbf{\bar{x}}_{t},\mathbf{y}_{t})}\mid \mathbf{y}_{t}\right] \times\nonumber\\
&&\displaystyle\prod_{t=t_{h}}^{n}\mathbb{E}\left[ 2^{\theta_{h}d_{t}(\mathbf{x}_{t},\mathbf{y}_{t}))}\mid \mathbf{y}_{t}\right], \text{w.p.1},
\end{eqnarray}
for BSC with uniformly distributed inputs and the decoding error cost measure given in \eqref{decodingcostwithdiscounts},
where
\begin{equation}
\mathbb{E}\left[ 2^{-\vartheta_{h;h'}d_{t}(\mathbf{\bar{x}}_{t},\mathbf{y}_{t})}\right]=\frac{1}{2}+\frac{1}{2}\left( \frac{1-p}{p}\right)^{-\vartheta_{h;h'}\gamma^{t-1}}, 
\end{equation}
\begin{equation}
\mathbb{E}\left[ 2^{\vartheta_{h;h'}d_{t}(\mathbf{x}_{t},\mathbf{y}_{t})}\right]=1-p+p\left( \frac{1-p}{p}\right)^{\vartheta_{h;h'}\gamma^{t-1}},
\end{equation}
\begin{equation}
\mathbb{E}\left[ 2^{-\theta_{h}d_{t}(\mathbf{\bar{x}}_{t},\mathbf{y}_{t})}\mid \mathbf{y}_{t}\right]=\frac{1}{2}+\frac{1}{2}\left( \frac{1-p}{p}\right)^{-\theta_{h}\gamma^{t-1}},\text{w.p.1},
\end{equation}
and,
\begin{equation}
\mathbb{E}\left[ 2^{\theta_{h}d_{t}(\mathbf{x}_{t},\mathbf{y}_{t}))}\mid \mathbf{y}_{t}\right]=1-p+p\left( \frac{1-p}{p}\right)^{\theta_{h}\gamma^{t-1}},\text{w.p.1}.
\end{equation}

For finding a good tree structure function $s(t)$ based on error bounds such as $D_{\text{E}}^{\text{(G)}}[s,d]$, we will a consider an algorithm that we refer to as \emph{successive bit placement} (SBP) algorithm. The SBP algorithm iteratively updates $s(t)$ starting from $s(t)=1, \forall t$ case such that each update places a new bit to a location  minimizing the increase in the error bound. For example, given a decoding error cost measure $d$, the SBP algorithm based on $D_{\text{E}}^{\text{(G)}}[s,d]$, updates $s(t)$ using the following update rule:
\begin{equation}
s_{k'+1}(t)= \arg\min_{s \in \mathcal{J}[s_{k'}]}D_{\text{E}}^{\text{(G)}}[s,d],
\end{equation}     
where
\begin{equation}
\mathcal{J}[s'] = \lbrace s : \exists j \leq n, s(t)=s'(t),  t < j, s(t)=s'(t)+1 ,  t \geq j \rbrace
\end{equation}
and
$s_{k'+1}(t)$ is the update for $s(t)$ at the step $k'+1$, i.e., $s(t)=s_{k'+1}(t)$ at the $k'+1$st update. Note that for evaluating $D_{\text{E}}^{\text{(G)}}[s,d]$ with $s \in \mathcal{S}[s_{k'}]$, $k$ terms \footnote{An example is $w_{h}=2^{k}\Pr(\tau = b_{h})$.} should be replaced with $k'+1$ as $s_{k'}$ corresponds to the tree structure function of an $(n,k'+1)$-random tree code, i.e., $s_{k'}(n)=k'+1$.

To apply the SBP algorithm based on $D_{\text{E}}^{\text{(G)}}[s,d]$ with exhaustive search for the minimization part, one can evaluate $D_{\text{E}}^{\text{(G)}}[s,d]$ $n$ times for each update hence  $(k-1)n$ evaluations of $D_{\text{E}}^{\text{(G)}}[s,d]$ are sufficient for placing all message bits after the first one.

Observe that the ensembles of random tree codes resulting from the application of the SBP algorithm based on $D_{\text{E}}^{\text{(G)}}[s,d]$, i.e., CORT codes , converge to the case where $s(1)=k$, i.e., pure random codes, as the computational limit $L$ grows large. To see this, consider the limiting case $L=\infty$ where $D_{\text{CLE}}^{\text{(G)}}[s,d]$ vanishes and $D_{\text{E}}^{\text{(G)}}[s,d]$ becomes $D_{\text{CFE}}^{\text{(G)}}[s,d]$. For $D_{\text{E}}^{\text{(G)}}[s,d]=D_{\text{CFE}}^{\text{(G)}}[s,d]$, the SBP algorithm always choose the earliest bit, i.e., the first bit, hence giving $s(1)=k$ as it finishes when all message bits are placed. 
Note that for the case with $\gamma=1$, the SSDGU returns the message having the ML codeword unless the computational limit is surpassed. Accordingly, it can be seen that $D_{\text{CFE}}^{\text{(G)}}[s,d]$ coincides with Gallager's bound for the case with $\gamma=1$ and $s(1)=k$, i.e., pure random codes.     
\section{Numerical Simulations}
\label{sec:Numerical_Simulations}
In this part, we will evaluate $D_{\text{E}}^{\text{(G)}}[s,d]$ considering $(128,64)$-random tree codes. To approximate the true value of $D_{\text{E}}^{\text{(G)}}[s,d]$, the parameters $\varrho$ and $\rho$ are optimized over a $10$-point uniform quantization of $[0,1]$.

\begin{table}[h]
\begin{center}
  \begin{tabular}{ | l | l | l | l | l |}
    \hline
                                     & $L=10^{9}$   & $L=10^{10}$   & $L=10^{11}$ \\ \hline
    $D_{\text{E}}^{\text{(G)}}[s,d]$ & $3.6 \times 10^{-3}$ & $1.9 \times 10^{-3}$ & $1.3 \times 10^{-3}$\\ \hline
    $D_{\text{CLE}}^{\text{(G)}}[s,d]$ & $1.7 \times 10^{-3}$ & $0.4 \times 10^{-3}$ & $0.8 \times 10^{-4}$\\ \hline
    $D_{\text{CFE}}^{\text{(G)}}[s,d]$ & $2.0 \times 10^{-3}$ & $1.5 \times 10^{-3}$ & $1.2 \times 10^{-3}$\\
    \hline
  \end{tabular}
\end{center}
\caption{The evaulation of $D_{\text{E}}^{\text{(G)}}[s,d]$, $D_{\text{CLE}}^{\text{(G)}}[s,d]$ and $D_{\text{CFE}}^{\text{(G)}}[s,d]$ for various computational limit $L$ values considereing $(128,64)$-random tree codes under BSC with $p=0.03$ where $s$ is optimized using the SBP algorithm and $d$ is in the form \eqref{decodingcostwithdiscounts} where $\gamma=1$. }
\label{table:DEfor(128,64)underp003gamma1}
\end{table}

\begin{table}[h]
\begin{center}
  \begin{tabular}{ | l | l | l | l | l |}
    \hline
                                     & $L=10^{9}$   & $L=10^{10}$   & $L=10^{11}$ \\ \hline
    $D_{\text{E}}^{\text{(G)}}[s,d]$ & $2.7 \times 10^{-3}$ & $1.7 \times 10^{-3}$ & $1.5 \times 10^{-3}$\\ \hline
    $D_{\text{CLE}}^{\text{(G)}}[s,d]$ & $0.6 \times 10^{-3}$ & $0.2 \times 10^{-3}$ & $0.7 \times 10^{-4}$\\ \hline
    $D_{\text{CFE}}^{\text{(G)}}[s,d]$ & $2.1 \times 10^{-3}$ & $1.5 \times 10^{-5}$ & $1.4 \times 10^{-3}$\\
    \hline
  \end{tabular}
\end{center}
\caption{The evaulation of $D_{\text{E}}^{\text{(G)}}[s,d]$, $D_{\text{CLE}}^{\text{(G)}}[s,d]$ and $D_{\text{CFE}}^{\text{(G)}}[s,d]$ for various computational limit $L$ values considereing $(128,64)$-random tree codes under BSC with $p=0.03$ where $s$ is optimized using the SBP algorithm and $d$ is in the form \eqref{decodingcostwithdiscounts} where $\gamma=0.9992$. }
\label{table:DEfor(128,64)underp003gamma09992}
\end{table}

\begin{table}[h]
\begin{center}
  \begin{tabular}{ | l | l | l | l | l |}
    \hline
                                     & $L=10^{9}$   & $L=10^{10}$   & $L=10^{11}$ \\ \hline
    $D_{\text{E}}^{\text{(G)}}[s,d]$ & $7.2 \times 10^{-5}$ & $2.6 \times 10^{-5}$ & $9.4 \times 10^{-6}$\\ \hline
    $D_{\text{CLE}}^{\text{(G)}}[s,d]$ & $3.7 \times 10^{-5}$ & $1.1 \times 10^{-5}$ & $2.8 \times 10^{-6}$\\ \hline
    $D_{\text{CFE}}^{\text{(G)}}[s,d]$ & $3.5 \times 10^{-5}$ & $1.4 \times 10^{-5}$ & $6.6 \times 10^{-6}$\\
    \hline
  \end{tabular}
\end{center}
\caption{The evaulation of $D_{\text{E}}^{\text{(G)}}[s,d]$, $D_{\text{CLE}}^{\text{(G)}}[s,d]$ and $D_{\text{CFE}}^{\text{(G)}}[s,d]$ for various computational limit $L$ values considereing $(128,64)$-random tree codes under BSC with $p=0.02$ where $s$ is optimized using the SBP algorithm and $d$ is in the form \eqref{decodingcostwithdiscounts} where $\gamma=1$. }
\label{table:DEfor(128,64)underp002gamma1}
\end{table}

\begin{table}[h]
\begin{center}
  \begin{tabular}{ | l | l | l | l | l |}
    \hline
                                     & $L=10^{9}$   & $L=10^{10}$   & $L=10^{11}$ \\ \hline
    $D_{\text{E}}^{\text{(G)}}[s,d]$ & $4.6 \times 10^{-5}$ & $1.7 \times 10^{-5}$ & $7.5 \times 10^{-6}$\\ \hline
    $D_{\text{CLE}}^{\text{(G)}}[s,d]$ & $2.2 \times 10^{-5}$ & $0.6 \times 10^{-5}$ & $1.8 \times 10^{-6}$\\ \hline
    $D_{\text{CFE}}^{\text{(G)}}[s,d]$ & $2.4 \times 10^{-5}$ & $1.1 \times 10^{-5}$ & $5.7 \times 10^{-6}$\\
    \hline
  \end{tabular}
\end{center}
\caption{The evaulation of $D_{\text{E}}^{\text{(G)}}[s,d]$, $D_{\text{CLE}}^{\text{(G)}}[s,d]$ and $D_{\text{CFE}}^{\text{(G)}}[s,d]$ for various computational limit $L$ values considereing $(128,64)$-random tree codes under BSC with $p=0.02$ where $s$ is optimized using the SBP algorithm and $d$ is in the form \eqref{decodingcostwithdiscounts} where $\gamma=0.9992$. }
\label{table:DEfor(128,64)underp002gamma09992}
\end{table}

In Table \ref{table:DEfor(128,64)underp003gamma1} and Table \ref{table:DEfor(128,64)underp003gamma09992}, we evaluate $D_{\text{E}}^{\text{(G)}}[s,d]$, $D_{\text{CLE}}^{\text{(G)}}[s,d]$ and $D_{\text{CFE}}^{\text{(G)}}[s,d]$ under BSC with $p=0.03$ setting where $s$ is optimized using the SBP algorithm and $d$ is in the form \eqref{decodingcostwithdiscounts} where $\gamma=1$ and $\gamma=0.9992$, respectively. For BSC with $p=0.03$, $D_{\text{CFE}}^{\text{(G)}}[s,d]$ can be minimized to $1.1 \times 10^{-3}$ with  pure random codes, i.e., $s(1)=64$ case, and setting $\gamma=1$. 

In Table \ref{table:DEfor(128,64)underp002gamma1} and Table \ref{table:DEfor(128,64)underp002gamma09992}, we evaluate $D_{\text{E}}^{\text{(G)}}[s,d]$, $D_{\text{CLE}}^{\text{(G)}}[s,d]$ and $D_{\text{CFE}}^{\text{(G)}}[s,d]$ under BSC with $p=0.02$ setting where $s$ is optimized using the SBP algorithm and $d$ is in the form \eqref{decodingcostwithdiscounts} where $\gamma=1$ and $\gamma=0.9992$, respectively. For BSC with $p=0.02$, $D_{\text{CFE}}^{\text{(G)}}[s,d]$ can be minimized to $2.9 \times 10^{-6}$ with  pure random codes, i.e., $s(1)=64$ case, and setting $\gamma=1$. 

Note that one can obtain an upper bound on the expected number of node checks, i.e., $\mathbb{E}[N_{c}]$ , as $\mathbb{E}[N_{c}] \leq  D_{\text{CLE}}^{\text{(G)}}[s,d]L$ (see Theorem \ref{expectedNcboundDCLEML}). For example, the expected number of node checks is lower that $2.2 \times 10^{4}$ where $L=10^{9}$ with the setting in Table \ref{table:DEfor(128,64)underp002gamma09992}.  

Comparing $\gamma=1$ and $\gamma=0.9992$ cases, it can be observed that the setting with $\gamma=0.9992$ provides computational efficiency and lower error bounds particularly when $L$ is relatively low. However, for higher values of $L$, the setting with $\gamma=1$ closes the gap and can provide even lower error bounds. The reason is that as $\gamma=1$ is optimal (corresponding to ML decoding) for $D_{\text{CFE}}^{\text{(G)}}[s,d]$, $D_{\text{E}}^{\text{(G)}}[s,d]$ for the setting with $\gamma=1$ converges to the minimum error bound for $D_{\text{CFE}}^{\text{(G)}}[s,d]$, i.e., Gallager's bound, as $L$ grows large.      

These numerical examples show that CORT codes can approach the ML decoding performance of pure random codes under reasonable computational costs for the decoding.  

\section{Conclusion}
\label{sec:Conclusion}
We introduced an achievability bound on the frame error rate of random tree codes considering a sequential decoding with a hard computational limit. We proposed the design of practical codes based on the optimization of branching structure of the random tree codes and the decoding measure with respect to this achievability bound. 

We suggested an algorithm for optimizing the branching structure, i.e., the SBP algorithm, however we did not suggest such a method for optimizing the decoding measure. The varitional methods might be applied for the optimization of the decoding measure.

Numerical examples show that a heustically optimized and relaxed version of the achievability bound is effective. Yet, one can consider the tighter version of the achievability bound and other optimization methods.

In terms of hardware costs, one drawback of the SSDGU algorithm is that it is stack based and assumes a stack memory as large as the computational limit parameter. However, one can consider a low memory decoding algorithm based on the SSDGU algorithm. 
     
\bibliographystyle{ieeetr}
\bibliography{cortcodes}

\begin{thebibliography}{10}

\bibitem{6773024}
C.~E. Shannon, ``A mathematical theory of communication,'' {\em The Bell System
  Technical Journal}, vol.~27, no.~3, pp.~379--423, 1948.

\bibitem{Wozencraft1957SequentialDF}
J.~M. Wozencraft, ``Sequential decoding for reliable communication,'' 1957.

\bibitem{Gallager1968InformationTA}
R.~G. Gallager, ``Information theory and reliable communication,'' 1968.

\bibitem{5075875}
E.~Arikan, ``Channel polarization: A method for constructing capacity-achieving
  codes for symmetric binary-input memoryless channels,'' {\em IEEE
  Transactions on Information Theory}, vol.~55, no.~7, pp.~3051--3073, 2009.

\bibitem{DBLP:journals/corr/abs-1908-09594}
E.~Arikan, ``From sequential decoding to channel polarization and back again,''
  {\em CoRR}, vol.~abs/1908.09594, 2019.

\bibitem{Cerezo2020VariationalQA}
M.~Cerezo, A.~Arrasmith, R.~Babbush, S.~C. Benjamin, S.~Endo, K.~Fujii, J.~R.
  McClean, K.~Mitarai, X.~Yuan, L.~Cincio, and P.~J. Coles, ``Variational
  quantum algorithms,'' {\em Nature Reviews Physics}, vol.~3, pp.~625 -- 644,
  2020.

\bibitem{Farhi2014AQA}
E.~Farhi, J.~Goldstone, and S.~Gutmann, ``A quantum approximate optimization
  algorithm,'' {\em arXiv: Quantum Physics}, 2014.

\bibitem{Zigangirov1975ProceduresOS}
K.~S. Zigangirov, ``Procedures of sequential decoding,'' 1975.

\bibitem{Jelinek1969FastSD}
F.~Jelinek, ``Fast sequential decoding algorithm using a stack,'' {\em Ibm
  Journal of Research and Development}, vol.~13, pp.~675--685, 1969.

\bibitem{5452208}
Y.~Polyanskiy, H.~V. Poor, and S.~Verdu, ``Channel coding rate in the finite
  blocklength regime,'' {\em IEEE Transactions on Information Theory}, vol.~56,
  no.~5, pp.~2307--2359, 2010.

\bibitem{8630851}
K.~R. Duffy, J.~Li, and M.~Médard, ``Capacity-achieving guessing random
  additive noise decoding,'' {\em IEEE Transactions on Information Theory},
  vol.~65, no.~7, pp.~4023--4040, 2019.

\bibitem{10619485}
Q.~Wang, Y.~Wang, Y.~Wang, J.~Liang, and X.~Ma, ``Random staircase generator
  matrix codes,'' in {\em 2024 IEEE International Symposium on Information
  Theory (ISIT)}, pp.~2622--2627, 2024.

\end{thebibliography}
\end{document}